\documentclass{aims}
\usepackage{amsmath}
  \usepackage{paralist}
  \usepackage{graphics} 
  \usepackage{epsfig} 
\usepackage{graphicx}  
\usepackage{bm}
 \usepackage[colorlinks=true]{hyperref}
\hypersetup{urlcolor=blue, citecolor=red}

\def\currentvolume{X}
 \def\currentissue{X}
  \def\currentyear{200X}
   \def\currentmonth{XX}
    \def\ppages{X--XX}
     \def\DOI{10.3934/xx.xx.xx.xx}

\newtheorem{theorem}{Theorem}[section]

\newtheorem{proposition}{Proposition}

\theoremstyle{definition}

\newtheorem{remark}{Remark}

\DeclareMathOperator{\supp}{supp}

\title[Micro solutions of the B-E equation] 
      {Microscopic solutions of the Boltzmann--Enskog equation in the series representation}

\author[Mario Pulvirenti, Sergio Simonella and Anton Trushechkin]{}


 \email{pulviren@mat.uniroma1.it}
 \email{sergio.simonella@ens-lyon.fr}
 \email{trushechkin@mi.ras.ru}

\begin{document}
\maketitle

\centerline{\scshape Mario Pulvirenti}
\medskip
{\footnotesize
 \centerline{International Research Center M\&MOCS, Universit\`a dell'Aquila}
   \centerline{Palazzo Caetani, Cisterna di Latina, (LT) 04012 Italy}
} 

\medskip

\centerline{\scshape Sergio Simonella}
\medskip
{\footnotesize
 \centerline{CNRS and UMPA (UMR CNRS 5669), \'{E}cole Normale Sup\'{e}rieure de Lyon}
   \centerline{46 all\'{e}e dÕItalie, 69364 Lyon Cedex 07, France}
}

\medskip

\centerline{\scshape Anton Trushechkin}
\medskip
{\footnotesize
 \centerline{Steklov Mathematical Institute of Russian Academy of Sciences}
 \centerline{Gubkina Street 8, Moscow 119991, Russia;}
   \centerline{National Research Nuclear University MEPhI}
      \centerline{Kashirskoe Highway 31, Moscow 115409, Russia;}
   \centerline{National University of Science and Technology MISIS}
      \centerline{Leninsky Avenue 2, Moscow 119049, Russia}
}

\bigskip


\begin{abstract}
The Boltzmann--Enskog equation for a hard sphere gas is known to have so called microscopic solutions, i.e., solutions of the form of time-evolving empirical measures of a finite number of hard spheres. However, the precise mathematical meaning of these solutions should be discussed, since the formal substitution of empirical measures into the equation is not well-defined. Here we give a rigorous mathematical meaning to the microscopic solutions to the Boltzmann--Enskog equation by means of a suitable series representation.
\end{abstract}

\section{Introduction}

The present paper is devoted to the Boltzmann--Enskog equation, which describes the kinetics of a hard sphere gas:
\begin{equation}\label{eq:bee}
\begin{split}
(\partial_t+v\cdot \nabla_x)&f(x,v,t) = \lambda\int_{\mathbb R^3\times S^2_+} dv_1 d\omega \ (v-v_1)\cdot\omega  \\
&\times \Big\{f(x-a\omega,v_1',t)f(x,v',t)- f(x+a\omega,v_1,t)f(x,v,t)\Big\}\;,
\end{split}
\end{equation}
where the unknown  $f=f(x,v,t)\geq0$ denotes the density function of the system, $x\in\mathbb R^3$, $v\in\mathbb R^3$ and $t\in\mathbb R$ denote position, velocity and time respectively. Moreover $a>0$ is the diameter of a hard sphere. If the function is normalized, i.e., $\int_{\mathbb R^6}f(x,v,t)\,dxdv=1$, then $f$ can be understood as a probability density of an arbitrary single hard sphere. Further, $S_+^2=\{\omega \in S^2 |\ (v-v_1)\cdot\omega \geq 0\},$ $S^2$ is the unit sphere in $\mathbb R^3$
(with the surface measure $d\omega$), $(v,v_1)$ is a pair of 
velocities in incoming collision configuration and $(v',v_1')$ is the corresponding pair of outgoing 
velocities defined by the elastic reflection rules (with $\omega$ being the unit vector directed from the center of the first sphere to the center of the second one)
\begin{equation}
\begin{cases}
\displaystyle v'=v-\omega [\omega\cdot(v-v_1)]\\
\displaystyle  v_1'=v_1+\omega[\omega\cdot(v-v_1)].
\end{cases}
\label{eq:coll}
\end{equation}
Finally, $\lambda>0$ is a parameter modulating the collision rate. The integral in (\ref{eq:bee}) is referred to as the collision integral.

This equation differs from the Boltzmann kinetic equation for hard spheres by the terms $\pm a\omega$ in the arguments of $f$ in the collision integral. Namely the Boltzmann equation assumes the size of spheres to be negligibly small (in comparison to the scale of spatial variation of $f$), while the Boltzmann--Enskog equation takes the size of spheres into account. Formally, the Boltzmann equation for hard spheres is obtained from the Boltzmann--Enskog equation in the limit $a\to0$. The convergence of solutions of the Boltzmann--Enskog equation to solutions of the Boltzmann equation is proved in \cite{ArkCerc89,ArkCerc,BelLach88}. 
Thus Eq.\,\eqref{eq:bee} is often regarded as a correction to the Boltzmann equation
which is both a simpler mathematical model and a kinetic description of a\,``dense gas''.

An important problem in the theory of the Boltzmann equation is its derivation from the microscopic dynamics (Hamiltonian particle system), see \cite{Bogol46,Spohn,CIP,Lanford} and, for more recent research, \cite{Bodi,GerasGap,GST,PulvirSimonBG,PSS,De17}. Up to now, the Boltzmann equation for hard spheres and other short-range potentials has been rigorously derived for short times only. 
In contrast to the Boltzmann equation, it is not clear whether an Enskog type equation can be rigorously derived from deterministic dynamics.

 It is interesting that, in case of hard spheres, Bogolyubov's microscopic derivation in \cite{Bogol46}  leads not to the Boltzmann equation but to the Boltzmann--Enskog equation (\ref{eq:bee}). However, the latter equation 
 does not provide a better approximation to the hard sphere dynamics, but just a natural intermediate 
 description \cite{PulvirSimonBG}.

Moreover, the Boltzmann--Enskog equation has an interesting property, which relates it to the microscopic dynamics in a different way. Namely, Bogolyubov discovered \cite{Bogol75} (see also \cite{BogBog}) that, for any finite $N$, the Boltzmann--Enskog equation with $\lambda=Na^2$ has solutions of the form of time-evolving empirical measures of $N$ hard spheres:

\begin{equation}\label{eq:microsol}
\mu_t(dz)=\frac1N\sum_{i=1}^N\delta(z-z_i(t))dz.
\end{equation}
Here $z=(x,v)\in\mathbb R^6$ is a point in the phase space and $z_i(t)=(x_i(t),v_i(t))$ is the phase point of the $i$-th hard sphere determined by free motion between collisions plus elastic pairwise reflections at distance $a$. Setting $\mathbf z(t)=(z_i(t))_{i=1}^N=\mathrm T_N(t)\mathbf z^0 $, and  $\mathbf z^0=(z_i^0)_{i=1}^N$ the time evolved and initial configuration respectively,
the flow $\mathbf z^0\to\mathbf z(t)$ is defined for almost all $\mathbf z^0$ with respect to the Lebesgue measure. In fact configurations leading to collisions of more than two hard spheres simultaneously, to grazing collisions or to infinite collisions in a finite time have zero measure \cite{Alexander,CGP,CIP}.

Solutions of type \eqref{eq:microsol} are referred to as microscopic solutions. It may be surprising that the Boltzmann--Enskog equation\,``contains''\,the $N$-particle dynamics in itself, while the kinetic equation is valid in the limit of large $N$ only.

Such solutions are even more surprising if we recall that the Boltzmann--Enskog equation describes the irreversible dynamics of the gas. Namely, consider the functional
\begin{equation}\label{eq:hfunc}
H(f)=\int_{\mathbb R^6} f\ln f\,dxdv+
\frac{\lambda}{2}\int_{\mathbb R^3}dx\int_{B(x,a)}\rho(x)\rho(y)\,dy,
\end{equation}
where $B(x,a)$ is the ball around $x$ and radius $a$, and $\rho(x)=\int_{\mathbb R^3}f(x,v)\,dv$ is the spatial density. This functional does not increase if $f$ is a solution of the Boltzmann--Enskog equation \cite{ArkCerc}. In contrast, microscopic solutions are reversible in time (i.e., a transformation $t\to-t$ and $v\to-v$ maps a measure of type (\ref{eq:microsol}) to another measure of this type, which is also a solution of (\ref{eq:bee})). Of course, functional (\ref{eq:hfunc}) is not defined on solutions of type (\ref{eq:microsol}), hence, formally, there is no contradiction. However, this shows that the reversibility or irreversibility of the Boltzmann--Enskog equation depends on the considered class of solutions: we have irreversible behaviour if we consider regular solutions, and reversible behaviour in the case of suitable singular measures (\ref{eq:microsol}).

A difficulty with the microscopic solutions (\ref{eq:microsol}) is that their formal substitution into (\ref{eq:bee}) yields products of delta functions and, hence, it is ill-defined. A rigorous sense to these solutions was given in \cite{TrushMIAN} by means of regularizations for delta functions and the collision integral (see also \cite{TrushPad,TrushKRM} for other variants).
On the other hand it was proven in \cite{PulvirSimon} that the empirical distributions \eqref{eq:microsol} (more precisely the family of the empirical marginals) solve the BBGKY hierarchy for hard spheres. Actually, 
Bogolyubov's remark in \cite{Bogol75} was that the Boltzmann--Enskog equation coincides formally with the first equation of the BBGKY hierarchy for hard spheres. 

As discussed in \cite{PulvirSimon}, the approach therein developed, based on the standard notion of Duhamel series solution, has no simple adaptation to \eqref{eq:bee}. In the present paper, we introduce a modified notion of series solution for which \eqref{eq:microsol} does solve the Boltzmann--Enskog equation. Such a notion is based on tree expansions with partially ordered trees, in contrast with the standard expansion on totally ordered trees. This, together with a regularization based on separation of collision times, allows to formulate our main result (Theorem \ref{th} in Section \ref{sec:MsBEE} below). 

 The paper is organized as follows. In Section \ref{sec:series} we introduce the trees and partially ordered trees together with the standard series solution. In Section \ref{sec:MVS} the concept of series solution for measures of type \eqref{eq:microsol} is precisely formulated. We establish also a semigroup property which will be crucial, in Section \ref{sec:MsBEE}, for the proof of Theorem \ref{th}. In the latter section we finally
compare the strategy of \cite{PulvirSimon} with the one of the present paper.

\section{Series solution and tree expansion}\label{sec:series}
\subsection{Series expansion and trees}

We will use the following notations: $\mathbf z_j=(z_1,\ldots,z_j)$, $z_i=(x_i,v_i)$,
$$f^{\otimes j}(\mathbf z_j,t)=f(z_1,t)\cdots f(z_j,t).$$
Also, for shortness, we will write $f(t)$ and  $f^{\otimes j}(t)$ for $f(z_1,t)$ and $f^{\otimes j}(\mathbf z_j,t)$ respectively.  Let us introduce the operators $\mathcal S_i(t)$ and $\mathcal C^{\pm}_{i,j+1}$:
\newpage
$$(\mathcal S_i(t)f^{\otimes j})(\mathbf z_j,t)=f(z_1,t)\cdots f(z_{i-1},t)
f(x_i-v_it,v_i,t)f(z_{i+1},t)\cdots f(z_j,t),$$
\begin{eqnarray*}
&&\left(\mathcal C^+_{i,j+1}f^{\otimes(j+1)}\right)(\mathbf z_j,t)= 
\int_{\mathbb R^3\times S^2_+} dv_{j+1} d\omega_{j+1} \ (v_i-v_{j+1})\cdot\omega_{j+1}  \\ &&\times f(z_1,t)\cdots f(z_{i-1},t)f(x_i,v_i',t)f(z_{i+1},t)\cdots f(z_j,t)f(x_i-a\omega_{j+1},v'_{j+1},t),
\end{eqnarray*}
\begin{eqnarray*}
&&\left( \mathcal C^-_{i,j+1} f^{\otimes (j+1)}\right)(\mathbf z_j,t)=\int_{\mathbb R^3\times S^2_+} dv_{j+1} d\omega_{j+1} \ (v_i-v_{j+1})\cdot\omega_{j+1}  \\
&&\times f(z_1,t)\cdots f(z_{i-1},t)f(x_i ,v_i)f(z_{i+1},t)\cdots f(z_j,t)f(x_i+a\omega_{j+1},v_{j+1},t),
\end{eqnarray*}
$$\mathcal C_{i,j+1}= \mathcal C^+_{i,j+1}- \mathcal C^-_{i,j+1}.$$
We remind that $S_+^2=\{\omega \in S^2 |\ (v-v_1)\cdot\omega \geq 0\}$  and  $S^2$ is the unit sphere in $\mathbb R^3$.

Then, the Boltzmann--Enskog equation (\ref{eq:bee}) can be rewritten as 
$$\left(\partial_t+v_1\cdot\nabla_{x_1}\right)f(t)=
\lambda \mathcal C_{1,2}f^{\otimes 2}(t),$$
or, in the integrated form, as
\begin{equation}\label{eq:beeint}
f(t)=\mathcal S_1(t)f_{0}+\lambda\int_0^tds\,\mathcal S_1(t-s)
\mathcal C_{1,2} f(s)f(s)\;,
\end{equation}
where $f_0\equiv f_0(z_1)$.

If we iterate  equation (\ref{eq:beeint}) (i.e., iteratively substitute (\ref{eq:beeint}) into the right-hand side of itself), we obtain a series solution in the form
\begin{multline}\label{eq:hier}
f_j(t)=\mathcal S_{1\ldots j}f_{0,j}+\sum_{n=1}^\infty\lambda^n
\int_0^tdt_1\int_0^{t_1}dt_2\ldots\int_0^{t_{n-1}}dt_n\\
\mathcal S_{1\ldots j}(t-t_1)\mathcal C_{j+1}
\mathcal S_{1\ldots j+1}(t_1-t_2)\ldots
\mathcal C_{j+n}\mathcal S_{1\ldots j+n}(t_n)f_{0,j+n},
\end{multline}
where $f_j (t)= f^{\otimes j}(t)$, $f_{0,j+n} = f_{0}^{\otimes (j+n)}$,
\begin{equation}\label{eq:cj}
\mathcal C_j=\sum_{l=1}^{j-1}C_{l,j},
\end{equation}and 
$\mathcal S_{1\ldots j}(s)=\mathcal S_{1}(s)\ldots\mathcal S_{j}(s)$. The substitution of (\ref{eq:cj}) into (\ref{eq:hier}) yields
\begin{multline}\label{eq:hier2}
f_j(t)=\mathcal S_{1\ldots j}f_{0,j}+\sum_{n=1}^\infty\lambda^n
{\sum_{\mathbf r_n}}^*
\int_0^tdt_1\int_0^{t_1}dt_2\ldots\int_0^{t_{n-1}}dt_n\\
\mathcal S_{1\ldots j}(t-t_1)\mathcal C_{r_1,j+1}
\mathcal S_{1\ldots j+1}(t_1-t_2)\ldots
\mathcal C_{r_n,j+n}\mathcal S_{1\ldots j+n}(t_n)f_{0,j+n},
\end{multline}
where 
$$
{\sum_{\mathbf r_n}}^*=\sum_{r_1=1}^j\sum_{r_2=1}^{j+1}\cdots
\sum_{r_n=1}^{j+n-1}.
$$
Expansion (\ref{eq:hier2}) leads to a natural tree representation, where $r_i$ is a ``parent'' of the particle $j+i$. (Again, if we are interested only on the single-particle distribution function, we should set $j=1$. Here we are introducing the `hierarchy' of equations for the functions $f_j$, which will be useful in the sequel.)

The collection of integers 
$$
\mathbf {r_n}= \{r_1, \cdots, r_n \}
$$
is called `tree'. The name is justified by the fact that a tree is conveniently represented graphically. For instance
the tree  $\mathbf {r_5}= \{ 1,1,2,3,2 \}$ for $j=1$ is given by Fig.\,\ref{fig:trees} where the $i$-th branch is  generated by particle $r_i$;  see also Figure \ref{fig:trees'}. 

\begin{figure}[h]
\centering
\includegraphics[width=4cm]{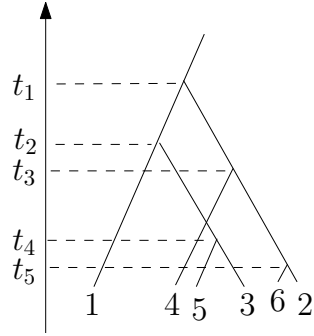}
\caption{: $\mathbf {r_5} = \{1,1,2,3,2\}$.}\label{fig:trees}
\end{figure}

\begin{figure}[h]
\centering
\includegraphics[width=\textwidth]{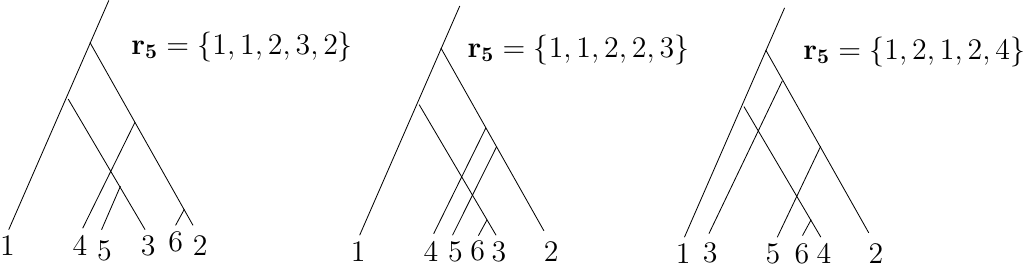} 
\caption{}\label{fig:trees'}
\end{figure}

\begin{remark}\label{rem:trees}
For the tree representation we refer to \cite{Spohn} (``Collision Histories''). Later on this notion has been used frequently by several authors \cite{Pulvir,GST, PSS,PulvirSimon,PulvirSimonBG,Simonella,BGSR13,BGSR15}.
 \end{remark}

 The series expansion \eqref{eq:hier2} can be further specified by splitting any operator ${\mathcal C}_j$ into its positive and negative part. The result is:
 \begin{multline}\label{eq:hier3}
f_j(t)=\mathcal S_{1\ldots j}f_{0,j}+\sum_{n=1}^\infty\lambda^n
{\sum_{\mathbf r_n}}^* \sum_{\bm\sigma_n} \prod_{i=1}^{n} \sigma_i
\int_0^tdt_1\int_0^{t_1}dt_2\ldots\int_0^{t_{n-1}}dt_n\\
\mathcal S_{1\ldots j}(t-t_1)\mathcal C^{\sigma_1}_{r_1,j+1}
\mathcal S_{1\ldots j+1}(t_1-t_2)\ldots
\mathcal C^{\sigma_n}_{r_n,j+n}\mathcal S_{1\ldots j+n}(t_n)f_{0,j+n},
\end{multline}
where $\bm\sigma_n=\{\sigma_1,  \cdots, \sigma_n\}$ and $\sigma_i= \pm$.

\subsection{Backward and forward Boltzmann--Enskog flows}\label{sec:flows}

The operators $C^{\sigma_i}_{r_i,j+i}$ are integrals over $d \omega_{i}$ and $dv_{j+i}$ which, together with $t_i$, will be called the `node variables' (see Fig.\,\ref{fig:trees}).

Given $n$, $\mathbf r_n$, $\bm\sigma_n$ and values of the node variables, we now define the so called \textit{Boltzmann--Enskog backward flow} for the configuration of particles $\bm\zeta(s)$, $s\in[0,t]$,
\begin{equation}\label{eq:weakstart}
\bm\zeta(s)=(\zeta_i(s))_{i\in\mathcal H(s)},\quad
\mathcal H(s)=\{1\}\cup\{i\geq2\,|\, t_{i-1}\geq s\},
\end{equation}
with $\zeta_i(s)=(\xi_i(s),\eta_i(s))$, respectively position and velocity of particle $i$. Let us define the flow. Note that the number of particles is increasing backward in time.

We are considering the case $j=1$ for simplicity.

Firstly, 
\begin{equation}
\bm\zeta(t)=\zeta_1(t)=(x_1,v_1)
\end{equation}
(recall that $(x_1,v_1)$ are the arguments of the function $f_1$ in (\ref{eq:hier3})). At the instant $t$, only the first particle is under consideration. It moves freely, back in time, up to the instant of the first creation $t_1$, i.e.
$$
\bm\zeta(t_1+0)=\zeta_1(t_1+0)=(x_1-v_1t_1,v_1).
$$ 
The time instant $t_1$ corresponds to a contact with the particle~2 with the configuration $(x_1-\sigma_1a\omega_1,v_1)$ (``contact'' means that the distance between the centres of the hard spheres is exactly $a$). We add the configuration of this particle to $\bm\zeta$. If $\sigma_1=-$, then the configuration of the particles is precollisional and both particles continue to move freely (backward in time) with their velocities $(v_1,v_2)$, so that
$$\bm\zeta(t_1-0)=(\zeta_1(t_1-0), \zeta_2(t_1-0))=
\big((x_1-v_1t_1,v_1),(x_1-v_1t_1+a\omega_1,v_2)\big).$$ 
If $\sigma_1=+$, then the configuration of the particles is postcollisional, and the particles continue to move freely (backward in time) with the precollisional velocities $(v'_1,v'_2)$:
$$\bm\zeta(t_1-0)=(\zeta_1(t_1-0), \zeta_2(t_1-0))=
\big((x_1-v_1t_1,v'_1),(x_1-v_1t_1-a\omega_1,v'_2)\big).$$ 
Both particles continue to move freely (backward in time) up to their next collision, and so on. Since $r_i$ is the ``parent particle'' of the particle~$i$, $i=2,\ldots,n+1$, the general formula reads
\begin{equation}
\zeta_i(t_{i-1}-0)=
\begin{cases}
(\xi_{r_i}(t_{i-1})+a\omega_{i-1},v_i),&\sigma_{i-1}=-,\\
(\xi_{r_i}(t_{i-1})-a\omega_{i-1},v'_i),&\sigma_{i-1}=+\;,
\end{cases}
\end{equation}
for $i=2,\ldots,n+1$.

The above prescription defines the backward Boltzmann--Enskog flow. 

Thus, given $t>0$, we have a $(\mathbf r_n,\bm\sigma_n)$-dependent map
\begin{equation}\label{eq:bflow}
(x_1,v_1,\mathbf t_n,\bm\omega_n,\mathbf v_n)\longmapsto \bm\zeta(0),
\end{equation}
where $\mathbf t_n=(t_1,\ldots,t_n)$, $\bm\omega_n=(\omega_1,\ldots,\omega_n)$, $\mathbf v_n=(v_2,\ldots,v_{n+1})$. 
This map is important since it appears in formula \eqref{eq:hier3}, by writing explicitly the transport and collision operators.

The transformation (\ref{eq:bflow}) is a Borel map and we will denote its image as 
$\tilde A_{\mathbf r_n\bm\sigma_n}(t)$. 
Obviously, the image for $n=0$ is $\mathbb R^6$. The Jacobian determinant of the transformation is in modulus 
$a^{2n}\prod_{i=1}^n| \omega_i \cdot (v_{i+1}-\eta_{r_i}(t_i+0))|$, so that the map induces the equivalence of measures

$$
dx_1dv_1d\Lambda\, a^{2n}\prod_{i=1}^n| \omega_i \cdot (v_{i+1}-\eta_{r_i}(t_i+0))|=
d\bm\zeta(0),
$$
where
\begin{equation*}
d\Lambda(\mathbf t_n,\bm\omega_n,\mathbf v_n)=\chi_{\mathbf r_n}(\mathbf t_n)
dt_1\ldots dt_nd\omega_1\ldots d\omega_ndv_2\ldots dv_{n+1},
\end{equation*}

\begin{equation*}
\chi_{\mathbf r_n}(\mathbf t_n)=
\begin{cases} 
1,&0< t_{i+1}< t_i \quad i=0,\ldots,n,\, t_0=t, t_{n+1}=0 \\
0,&\text{otherwise.}
\end{cases}
\end{equation*}

We will use also the inverse of map  (\ref{eq:bflow}). To construct it, we introduce the \textit{Boltzmann--Enskog forward flow}, dependent on $(\mathbf k_n,\bm\sigma_n)$:
\begin{equation}
\label{map}
\bm\zeta(0)\longmapsto\left(\bm\zeta^F(s)\right)_{s\in[0,t]}=
\left(\bm\zeta^F(s,\bm\zeta(0),\mathbf r_n,\bm\sigma_n)\right)_{s\in[0,t]},
\end{equation}
where $\bm\zeta(0)\in \tilde A_{\mathbf r_n\bm\sigma_n}(t)$. A tree or, algebraically, the $n$-tuple $\mathbf r_n$ prescribes the sequences of creations (i.e., collisions) of the $n$ particles. 

The Boltzmann--Enskog forward flow is defined as follows. A particle $i$ moves freely up to a contact with some other particle $j$. If the particle~$j$ is not the next collision partner of the particle~$i$ or the particle~$i$ is not the next collision partner of the particle~$k$, then they ignore each other, i.e., they go through each other reaching a mutual distance $\leq a$. Otherwise, if the two particles are the next collision partners as specified by $\mathbf r_n$, then the particle with the larger number disappears from $\bm\zeta^F$. For definiteness, let $i>k$, hence $k=r_i$. If $\sigma_{i-1}=-$, then the particle~$k$ does not change its velocity. If $\sigma_{i-1}=+$, then the particle~$k$ changes its velocity according to law (\ref{eq:coll}) (with $v$ and $v_1$ substituted by the precollisional velocities of the particles~$k$ and $i$ respectively, and $\omega$ being the unit vector directed from the center of the hard sphere~$i$ to the center of the hard sphere~$k$). If, according to the tree, the particles~$i$ and~$k$ are the next collision partners, then the existence of a moment of their contact is guaranteed by the condition $\bm\zeta(0)\in \tilde A_{\mathbf r_n\bm\sigma_n}(t)$. 

By construction, if we apply the backward flow (\ref{eq:bflow}) to $(x_1,v_1,\mathbf t_n,\bm\omega_n,\mathbf v_n)$ with $s=0$ to obtain $\bm\zeta(0)$ and then apply the forward flow to $\bm\zeta(0)$, we will obtain $\bm\zeta^F(t)=\zeta^F_1(t)=(x_1,v_1)$.

\begin{remark}
It may be worth to underline that forward and backward flows are not directly connected with the real dynamics of a hard-sphere system. In these flows particles can overlap and the tree $\mathbf r_n$ specifies which pair of particles must collide once at contact and which particles overlap freely.
This formalism is just a way to represent the solution of a partial differential equation. A connection with the hard-sphere motion will be discussed later on.
\end{remark}

\subsection{Partially ordered trees}

Preliminary to the construction of measure valued solutions discussed in the next section, we introduce here a rearrangement of the previous series expansion, based on a different notion of tree which will be called `partially ordered tree' in contrast with the fully ordered tree -or simply `tree'- defined above. Here we disregard the mutual ordering of collision times, except for those particles having the same parent. Equivalently, we ignore the time ordering of the birthdays of descendants of different parents.

For instance consider the tree $ \mathbf r_5=\{1,1,2,3,2\}$ in Fig.\,\ref{fig:trees}, and the other two trees in Fig.\,\ref{fig:trees'}, in which the time ordering $t_5 <t_4$ and $t_3 <t_2$ is ignored. All these trees are equivalent to the first one as unlabelled graphs. We call then `partially ordered tree' a tree in which only the ordering of branches generated by a given one is fixed.
In  Fig.\;3 a partially ordered tree is the collection of a fully ordered trees with different ordering of the times $t_2, t_3$ and $t_4, t_5$ and so on.

\begin{figure}[h]
\centering
\includegraphics[width=5cm]{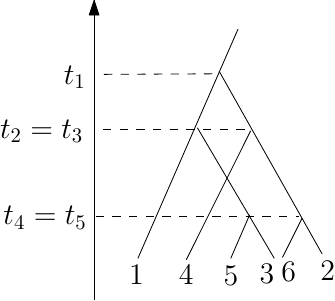}\label{fig:trees3}
\caption{}\label{fig:trees''}
\end{figure}

Each partially ordered tree is fully specified by the sequence of integers 
$$
\mathbf {k_n}= \{k_1, \cdots, k_n \}
$$
where $k_i$ is the number of collisions of the particle~$i$ on the way to its final point (after its creation). 
In particular, $k_1$ is the number of particles generated by the first particle.
The name assigned to each particle is specified as follows.

Defining
$$
K_i=\begin{cases}
1,&i=1,\\
1+\sum_{l=1}^{i-1}k_l, &i>1;
\end{cases}
$$
we fix the following ordering: particles created by particle $i$ obtain the numbers $K_i+1,\ldots,K_i+k_i$.We also set $k_{n+1}:=0$, 
since the particle with the last number does not suffer collisions.

We denote the set of partially ordered trees by $\mathcal K_n=\{\mathbf k_n=(k_1,\ldots,k_n)\}\subset \{0,\ldots,n\}^n$. We will use that it is given by $n$-tuples of integers such that
\begin{equation}
\sum_{i=1}^nk_i=n\;.
\end{equation}
The variables $k_i$ are related by the fact that the tree $\mathbf k_n$ must be realizable as graph. 
For instance the sequence $(1,0,2,1)$ with $n=4$ is not admissible, because it contains descendants without parents. 
First, we require  $k_1\geq1$ always. Then if $k_i >0$, particle $i$ must be already created in the procedure described above.

Note that we can represent the solution of the Boltzmann--Enskog equation in terms of partially ordered trees as 
\begin{equation}\label{eq:series}
f(t)=\sum_{n=0}^\infty\lambda^{n}f^{(n)}(t),
\end{equation}
\begin{subequations}\label{eq:fn}
\begin{eqnarray}
&f^{(0)}(t)=&\mathcal S_1(t)f_{0},\\
&f^{(n)}(t)=&\sum_{\mathbf k_n\in\mathcal K_n}\mathcal Q_{\mathbf k_n}(1;t)f_0^{\otimes(n+1)},\quad n\geq1,\label{eq:fnb}
\end{eqnarray}
\end{subequations}
where the operators $\mathcal Q$ are given by the following recursive formula:
\begin{eqnarray}
&&\mathcal Q_{\mathbf k_n}(i;s)=\nonumber\\&&
\left[\prod_{j=0}^{k_i-1}\int_0^{T_{ij}}dt_{K_i+j}
\mathcal S_i(T_{ij}-t_{K_i+j})\mathcal C_{i,K_i+j+1}
\mathcal Q_{\mathbf k_n}(K_i+j+1;t_{K_i+j})\right]\nonumber
\\&&
\times
\mathcal S_i(t_{K_i+k_i-1})\label{eq:q}
\end{eqnarray}
if $k_i\geq1$ and $\mathcal Q_{\mathbf k_n}(i;s)=\mathcal S_i(s)$ if $k_i=0$. 
Here
$$
T_{ij}=\begin{cases}
s,&j=0,\\
t_{K_i+j-1},&j>0.
\end{cases}
$$
\begin{figure}[h]
\centering
\includegraphics[width=6cm]{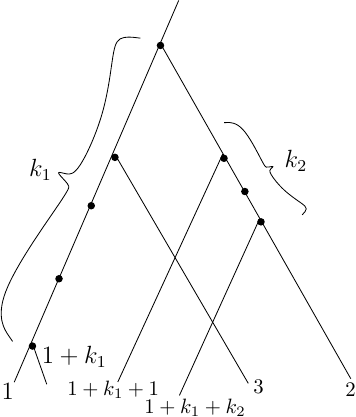}
\caption{}
\end{figure}
Then
\begin{equation*}
\begin{split}
\mathcal Q_{\mathbf k_n}(1;t) &= 
 \int_0^t  \mathcal S_1(t-t_1)\mathcal C_{1,2} \mathcal Q_{\mathbf k_n}(2;t_1)
\int_0^{t_1}  \mathcal S_1(t_1-t_2)\mathcal C_{1,3} \mathcal Q_{\mathbf k_n}(3;t_2)\times  \cdots \\
 &\times\int_0^{t_{k_1-1} }  \mathcal S_1(t_{k_1-1}-t_{k_1})\mathcal C_{1,k_1} 
\mathcal Q_{\mathbf k_n}(k_1;t_{k_1-1} ) \, \mathcal S_1(t_{k_1} ) \;.
\end{split}
\end{equation*}
For each $ \mathcal Q_{\mathbf k_n}(\ell;t_{\ell-1} )$ we repeat the procedure to arrive to Eq.\,\eqref{eq:series}.

Observe that the notions of backward and forward Boltzmann--Enskog flow, introduced in the previous section, can be extended easily to the present context.
We define as before the $(\mathbf k_n,\bm\sigma_n)$-dependent map
\begin{equation}\label{eq:bflow1}
(x_1,v_1,\mathbf t_n,\bm\omega_n,\mathbf v_n)\longmapsto \bm\zeta(0),
\end{equation}
the only difference being that the times $\mathbf t_n=(t_1,\ldots,t_n)$ are only partially ordered. Moreover we denote by
$ A_{\mathbf k_n\bm\sigma_n}(t)$ the image of this Borel map. 

Notice that two fully ordered trees yield the same partially ordered tree if they 
are equivalent as topological (unlabelled) graphs. Therefore $ \mathbf k_n$ can be thought as an equivalence class of fully ordered trees and we write $\mathbf r_n \in \mathbf k_n$ if  $\mathbf r_n$ belongs to the equivalence class specified by
$\mathbf k_n$. We have
$$
\bigcup_{\mathbf r_n\in\mathbf k_n}
\widetilde A_{\mathbf r_n\bm\sigma_n}=
A_{\mathbf k_n\bm\sigma_n}.
$$
Moreover
$$
{\sum_{\mathbf r_n}}^*=
\sum_{\mathbf k_n\in\mathcal K_n}
\sum_{\mathbf r_n\in\mathbf k_n}.
$$

The introduction of the backward Boltzmann--Enskog flow allows to write the series solution in a more explicit way, namely 
\begin{multline}\label{eq:expl}
f(x_1,v_1,t)\\=
\sum_{n=0}^\infty\lambda^n
{\sum_{\mathbf k_n}}
\sum_{\bm\sigma_n\in\{\pm \}^n}
\int \, d\Lambda \prod_{i=1}^n\sigma_i
[\omega_i \cdot (v_{i+1}-\eta_{r_i}(t_i+0))]
f_0^{\otimes(n+1)}(\bm\zeta(0)),
\end{multline}
where $r_i$ is the parent of particle $i$ in the partially ordered tree, $\zeta(0)=\zeta(0,\mathbf k_n,\bm\sigma_n)$,
\begin{equation*}
d\Lambda(\mathbf t_n,\bm\omega_n,\mathbf v_n)=\chi_{\mathbf k_n}(\mathbf t_n)
dt_1\ldots dt_nd\omega_1\ldots d\omega_ndv_2\ldots dv_{n+1},
\end{equation*}
$\chi_{\mathbf k_n}(\mathbf t_n)$ is the indicator function of the partial order dictated by the tree
$\mathbf k_n$ and $\bm\zeta(s,\mathbf k_n,\bm\sigma_n)$, $ s \in (0,t)$ is the backward flow.
By using the change of variables induced by the map \eqref{eq:bflow1}, we arrive to
\begin{multline}\label{eq:expl1}
\int dz_1f(x_1,v_1,t)\phi(z_1)\\=
\sum_{n=0}^\infty\left(\frac\lambda{a^2}\right)^n
{\sum_{\mathbf k_n}}
\sum_{\bm\sigma_n\in\{\pm \}^n}\prod_{i=1}^n\sigma_i
\int_{ A_{\mathbf k_n  \bm\sigma_n}(t)} d\bm z_{n+1}
f_0^{\otimes (n+1)}( \bm z_{n+1}) \phi(\zeta_1(t))\;,
\end{multline}
for any bounded continuous function $\phi: \mathbb R^6\to \mathbb R$.

A fundamental property of the series expansion is the semigroup property.  Denote
$\mathcal T(t)$ the evolution operator sending $f_0$ to $f(t)$, as given by the above formulas.
Then 
$f(t)=\mathcal T(t)f_0
$.
If $\tau_1,\tau_2\geq0$, then by algebraic manipulations it follows that
\begin{equation}\label{eq:semigroup}
\mathcal T(\tau_1+\tau_2)=\mathcal T(\tau_2)\mathcal T(\tau_1).
\end{equation}
In Section \ref{sec:MVS}, more attention will be paid to the proof of the semigroup property for measure valued solutions, which will be needed to construct the microscopic solutions of the Boltzmann--Enskog equation.

\subsection {Convergence}

The first relevant property of the expansion in terms of partially ordered trees is that their number is much smaller than the number of fully ordered trees. Indeed, while
$$
\sum_{\mathbf r_n} 1=n!,
$$
$$
\sum_{\mathbf k_n\in\mathcal K_n}1 \leq  \sum_{\begin{smallmatrix}k_1,\ldots,k_n\geq0,\\
k_1+\ldots+k_n=n\end{smallmatrix}}1 = \sum_{\begin{smallmatrix}k_1,\ldots,k_n\geq0,\\
k_1+\ldots+k_n=n\end{smallmatrix}} 2^n \prod_{i=1}^n 2^{-k_i} <4^n\;.
$$
This allows to prove

\begin{proposition}\label{prop:converg}
If $\lambda<a^2/8\|f_0\|$ (where $\|\cdot\|$ is the $L^1$-norm), the series (\ref{eq:expl1}) is convergent (in $L^1(\mathbb R^6)$) for all $t>0$. 
\end{proposition}
\begin{proof}
We have
\begin{equation*}
\|f(t)\| \leq\mathcal  \|f_0\|+
\sum_{n=1}^\infty\left(\frac{\lambda}{a^2}\right)^n
\sum_{\mathbf k_n\in\mathcal K_n}
\sum_{\bm\sigma_n\in\{\pm\}^n}
\int_{A_{\mathbf k_n\bm\sigma_n}(t)}
f_0^{\otimes(n+1)}(\bm\zeta_{n+1})\,d\bm\zeta_{n+1}\;.
\end{equation*}
The generic term is smaller than
\begin{equation*}
\begin{split}
&\left(\frac{\lambda}{a^2}\right)^n\|f_0\|^{n+1}\sum_{\mathbf k_n\in\mathcal K_n}
\sum_{\bm\sigma_n\in\{\pm \}^n}1=
\left(\frac{2\lambda}{a^2}\right)^n\|f_0\|^{n+1}\sum_{\mathbf k_n\in\mathcal K_n}1
\\
\leq&\left(\frac{2\lambda}{a^2}\right)^n\|f_0\|^{n+1}
\sum_{\begin{smallmatrix}k_1,\ldots,k_n\geq0,\\
k_1+\ldots+k_n=n\end{smallmatrix}}1<
\left(\frac{8\lambda}{a^2}\right)^n\|f_0\|^{n+1}.
\end{split}
\end{equation*}
We see that the series geometrically converges whenever $\lambda<a^2/8\|f_0\|$.
\end{proof}

This slightly improves the result of \cite{Pulvir} due to another algebraic way of description of the same graphical tree structure: we specify a tree just by an $n$-tuple $\mathbf k_n$. 

Moreover it may be interesting to observe that one can prove the convergence for any $\lambda$, but small $t$ . However, this is not a real restriction, since the free energy functional  (\ref{eq:hfunc}) provides a uniform (in time) 
estimate of absolute continuity of the distribution function which allows to iterate the procedure to reach arbitrary times, see \cite{Pulvir}.

\section {Measure valued solutions} \label{sec:MVS}

To prove the existence of microscopic solutions to the Enskog equation, we need to extend the series solution discussed so far to the case of initial singular measures.
Proceeding formally we write, for any bounded continuous
$\varphi(x_1,v_1)$  on $\mathbb R^6$,  
\begin{multline}\label{eq:weak}
\int_{\mathbb R^6}\varphi(x_1,v_1)\mu_t(dx_1dv_1)\\=
\sum_{n=0}^\infty\left(\frac\lambda{a^2}\right)^n
\sum_{\mathbf k_n\in\mathcal K_n}
\sum_{\bm\sigma_n\in\{\pm \}^n}\prod_{i=1}^n\sigma_i
\int_{A_{\mathbf k_n\bm\sigma_n}(t)} \varphi\left(z_1^F(t)\right)
\mu_0^{n+1}(d\mathbf z_{n+1})
\end{multline}
where  $\mathbf z^F(t,\mathbf z_{n+1}, \mathbf k_n,\bm\sigma_n))$ is the forward Boltzmann--Enskog flow,  $\mu_0(dx_1dv_1)$ on $\mathbb R^6$ is any finite Borel measure and
$\mu_0^{n}(d\mathbf z_n):=\mu_0(dz_1)\ldots\mu_0( dz_n).$ 
We rewrite here the integration variables $\bm\zeta(0)$ as $\mathbf z_{n+1}$, and $\zeta_1^F(t)$ as $z_1^F(t)=z_1^F(t,\mathbf z_{n+1},\mathbf k_n,\bm\sigma_n)$.

\subsection{Pathologies}

It may arise an ambiguity in this definition  if  the $\mu_0^{n+1}$-measure of the boundary $\partial A_{\mathbf k_n\bm\sigma_n}(t)$ of the region $A_{\mathbf k_n\bm\sigma_n}(t)$ is not vanishing. Recall that $A_{\mathbf k_n\bm\sigma_n}(t)$ is the set of the initial $n+1$-particle configurations leading to the tree $(\mathbf k_n,\bm\sigma_n)$. So, $\partial A_{\mathbf k_n\bm\sigma_n}(t)$ is the set of the initial $n+1$-particle configurations that, depending on small perturbation, may realize or not  the tree $(\mathbf k_n,\bm\sigma_n)$. A configuration $\mathbf z_{n+1}$ belongs to 
$\partial A_{\mathbf k_n\bm\sigma_n}(t)$ in several cases. 

The first one corresponds to an initial configuration leading to a collision of the particles~1 and~2 exactly at the final instant $t$. 

If a configuration leading to a collision at the instant $t$ has a positive $\mu_0$-measure, then we need to consider either $t-\eta$ or $t+\eta$ instead. For this reason, we may demand (\ref{eq:weak}) to be satisfied not for all $t$ but for almost all $t$.

The second case when 
$\mathbf z_{n+1}\in\partial A_{\mathbf k_n\bm\sigma_n}(t)$
corresponds to initial configurations leading, in the forward Boltzmann--Enskog flow, to so called grazing collisions (collisions with $\omega\cdot(v-v_1)=0$ in (\ref{eq:coll})). 
However, since we are interested in initial data of the form \eqref{eq:microsol} for $t=0$ also such pathology can be avoided by simply using that the initial configuration does not deliver grazing collisions in a finite time. In fact the set of such configurations is of full Lebesgue measure.

The previous pathologies depend on the hard-sphere dynamics, but there is also a pathology depending on the structure of Eq.\,\eqref {eq:weak}. We illustrate it by means of an example. Consider the case $N=2$ with initial measure
$$
\mu_0(d\zeta_1 d\zeta_2) =\frac 12 ( \delta (\zeta_1 -z_1) +  \delta (\zeta_2 -z_2))d\zeta_1d\zeta_2
$$
and with initial state $(z_1,z_2)$ leading to a collision in the time $(0,t)$. If we compute the terms in the expansion relative to $n=0,1$, we easily recover the time evolved measure $\mu_t$. 
This means that particle $2$ is created in some definite instant $t_1$, to fit with the initial configuration. However, as pointed out in \cite{PulvirSimon}, there are other non vanishing contributions. For instance particle $3$ can be created by particle $2$ with $\sigma_2=-$ at time $t_1-0$. The initial contribution is
$$
\delta (\zeta_1 -z_1) \delta (\zeta_2 -z_2)\delta(\zeta_3 -z_1)=\delta (\zeta_1 -z_1) \delta (\zeta_2 -z_2)
\delta (\zeta_3 -\zeta_1)\;.
$$
Furthermore, an arbitrary number of branches with $\sigma_i  =-, i>1$ can be created accumulating at the first node, see Fig.\,\ref{fig:path}.
In order to prevent such an unphysical event, we introduce next a suitable notion of measure valued solution.

\begin{figure}[h]
\centering 
\includegraphics[width=\textwidth]{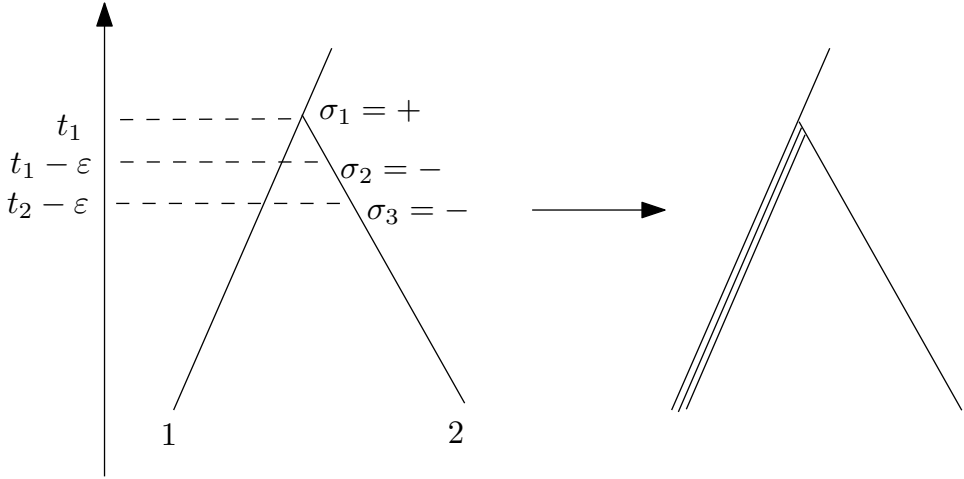}
\caption{: Initial configuration 
$\delta(\zeta_1-z_1)\delta(\zeta_2-z_2)\delta(\zeta_3-z_1)\delta(\zeta_4-z_1)\cdots$.
} \label{fig:path}
\end{figure}

\subsection{Regularization via time separation}

In order to overcome the pathologies discussed above, we modify
definition \eqref{eq:weak} by setting
\begin{multline}\label{eq:weakr}
\int_{\mathbb R^6}\varphi(x_1,v_1)\mu_t(dx_1dv_1)\\=
\sum_{n=0}^\infty\left(\frac\lambda{a^2}\right)^n
\sum_{\mathbf k_n\in\mathcal K_n}
\sum_{\bm\sigma_n\in\{\pm \}^n}\prod_{i=1}^n\sigma_i
\lim_{\varepsilon \to 0^+}
\int_{A^\varepsilon_{\mathbf k_n\bm\sigma_n}(t)} \varphi\left(z_1^F(t)\right)
\mu_0^{n+1}(d\mathbf z_{n+1})\;.
\end{multline}
The difference between \eqref{eq:weak} and \eqref{eq:weakr} consists in the introduction of the parameter $\varepsilon$
and in the replacement of $A_{\mathbf k_n\bm\sigma_n}(t)$ by $A^\varepsilon_{\mathbf k_n\bm\sigma_n}(t)$ defined by
$$
A^\varepsilon_{\mathbf k_n\bm\sigma_n}(t)=A_{\mathbf k_n\bm\sigma_n}(t)\backslash S^\varepsilon_{\mathbf k_n\bm\sigma_n}(t)\;,
$$
where the set
$
S^\varepsilon_{\mathbf k_n\bm\sigma_n}(t)
$
consists of all the elements $\mathbf z_{n+1}$ delivering, in the Boltzmann--Enskog forward flow, the event in which two particles,
born from the same progenitor, are created at times $t_i, t_j $, such that  $|t_i-t_j | \leq \varepsilon$. In other words, creations of particles from the same branch are time separated by $\varepsilon >0$.
Clearly, this avoids the main pathology envisaged in the previous section. 

We take Equation \eqref{eq:weakr} as definition of `regularized  series solution' for a given initial measure $\mu_0$. We will apply this notion only to initial measures which are empirical distributions, namely of the form  \eqref{eq:microsol} at time $0$. Notice, however, that the regularized series solution makes sense for any initial $\mu_0$.

\begin{remark} One could have defined the time separation on all the nodes of a fully ordered tree. This would not be good for our purpose, since simultaneous creations from different branches of a tree play a crucial role in the reconsruction of the hard-sphere dynamics (see Section \ref{subsec:RvC} below).
\end{remark}

A crucial property of the weak series solution is the semigroup property which we are going to illustrate.

\subsection{Semigroup property}

\begin{proposition}[Semigroup property]\label{prop:semigroup}
Let $\mu_0$ be such that configurations leading to grazing collisions have zero $\mu_0^{n+1}$-measure for all $n$, $\mathbf k_n$, and $\bm\sigma_n$. Let also $t>0$ and $\tau\in(0,t)$ be such that initial configurations leading to collisions at the instants $\tau$ or $t$ have zero $\mu_0^{n+1}$-measure for all $n$, $\mathbf k_n$, and $\bm\sigma_n$. If the following equalities are satisfied:

\begin{multline}\label{eq:tau}
\int_{\mathbb R^6}\varphi(x_1,v_1)\mu_\tau(dx_1dv_1)\\=
\sum_{n=0}^\infty\left(\frac\lambda{a^2}\right)^n
\sum_{\mathbf k_n\in\mathcal K_n}
\sum_{\bm\sigma_n\in\{\pm \}^n}\prod_{i=1}^n\sigma_i
\lim_{\varepsilon \to 0^+}
\int_{A^\varepsilon_{\mathbf k_n\bm\sigma_n}(\tau)} \varphi\left(z_1^F(\tau)\right)
\mu_0^{n+1}(d\mathbf z_{n+1})\;,
\end{multline}
\begin{multline}\label{eq:ttau}
\int_{\mathbb R^6}\varphi(x_1,v_1)\mu_t(dx_1dv_1)\\=
\sum_{n=0}^\infty\left(\frac\lambda{a^2}\right)^n
\sum_{\mathbf k_n\in\mathcal K_n}
\sum_{\bm\sigma_n\in\{\pm \}^n}\prod_{i=1}^n\sigma_i 
\lim_{\varepsilon \to 0^+}
\int_{A^\varepsilon_{\mathbf k_n\bm\sigma_n}(t-\tau)} \varphi\left(z_1^F(t-\tau)\right)
\mu_{\tau}^{n+1}(d\mathbf z_{n+1})\;,
\end{multline}
then equality (\ref{eq:weakr}) is also satisfied. Conversely, if (\ref{eq:weakr}) is satisfied for  $t$ and $\tau$, then  (\ref{eq:tau}) is also satisfied.
\end{proposition}

\begin{proof}

Rewrite the integral in  (\ref{eq:ttau}) as
\begin{multline}\label{eq:ttau2}
\int_{A^\varepsilon_{\mathbf k_n\bm\sigma_n}(t-\tau)}
\varphi\left(z_1^F(t-\tau,\mathbf z_{n+1},\mathbf k_n,\bm\sigma_n)\right)
\mu_\tau^{n+1}(d\mathbf z_{n+1})\\=
\int_{\mathbb R^6}\mu_\tau(dz_1)
\ldots
\int_{\mathbb R^6}\mu_\tau(dz_{n+1})
\psi(t-\tau,z_1,\ldots, z_{n+1},\mathbf k_n,\bm\sigma_n),
\end{multline}
where
$$
\psi(t-\tau,\mathbf z_{n+1},\mathbf k_n,\bm\sigma_n)=
\varphi\left(z_1^F(t-\tau,\mathbf z_{n+1},\mathbf k_n,\bm\sigma_n)\right) \chi_{A^\varepsilon_{\mathbf k_n\bm\sigma_n}(t-\tau)}(\mathbf z_{n+1}),
$$
and $\chi_{\mathcal A}$ is a characteristic function of a set $\mathcal A$. 

Firstly, we give a heuristic proof of the proposition, then we give a formal proof. We apply formula (\ref{eq:tau}) $n+1$ times. Each $n$-tuple $\mathbf k_n$ in expansion (\ref{eq:ttau}) corresponds to one tree for the interval $[\tau,t)$. Applications of (\ref{eq:tau}) to (\ref{eq:ttau2}) give $n+1$ trees $\mathbf k^1_{m_1},\ldots,\mathbf k^{n+1}_{m_{n+1}}$, which are, graphically, continuations of the final points of the tree $\mathbf k_n$ (see Fig.~\ref{fig:sumtrees}). So, the composition of the tree $\mathbf k_n$ with the trees $\mathbf k^1_{m_1},\ldots,\mathbf k^{n+1}_{m_{n+1}}$ yields a new tree $\mathbf K_N$, where $N=n+m_1+\ldots+m_{n+1}$. The Boltzmann--Enskog backward flows produced by the trees $\mathbf k^1_{m_1},\ldots,\mathbf k^{n+1}_{m_{n+1}}$ for the time interval $[0,\tau]$ together constitute a continuation of the  Boltzmann--Enskog backward flow produced by the tree $\mathbf k_n$. Their composition yields the Boltzmann--Enskog backward flow produced by the tree $\mathbf K_N$. Conversely, the Boltzmann--Enskog forward flow produced by the tree $\mathbf k_n$ is a continuation of the $n$-tuple of the Boltzmann--Enskog forward flows produced by the trees $\mathbf k^1_{m_1},\ldots,\mathbf k^{n+1}_{m_{n+1}}$. Their composition also yields the Boltzmann--Enskog forward flow produced by the tree $\mathbf K_N$. 

From the graphical representation, one can understand that the summations over all $\mathbf k_n$ and over all $\mathbf k^1_{m_1},\ldots,\mathbf k^{n+1}_{m_{n+1}}$ gives the summation over all joint trees $\mathbf K_N$ for the whole interval $[0,t]$. 

\begin{figure}[h]
\centering
\includegraphics[width=.75\textwidth]{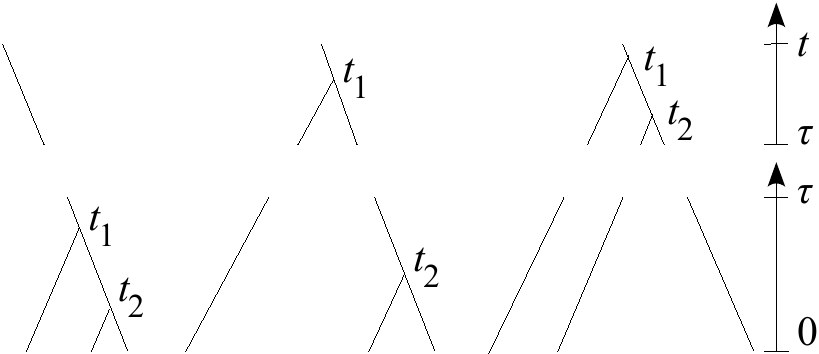}
\caption{Sum of compositions of trees.} 
\label{fig:sumtrees}
\end{figure}

Now we give a formal proof of the proposition. After the application of (\ref{eq:tau}) to (\ref{eq:ttau2}), a generic term is

\begin{multline}\label{eq:ttau3}
\left(\frac\lambda{a^2}\right)^N
\prod_{i=1}^{N}\Sigma_i
\lim_{\varepsilon\to0}
\int_{A^\varepsilon_{\mathbf k^1_{m_1}\bm\sigma^1_{m_1}}(\tau)\times\ldots\times
A^\varepsilon_{\mathbf k^{n+1}_{m_{n+1}}\bm\sigma^{n+1}_{m_{n+1}}}(\tau)}\\
\psi\left(t-\tau,\mathbf z_{n+1}^F(\tau,\mathbf z_{N+1},\mathbf K_N,\bm\Sigma_N),
\mathbf k_n,\bm\sigma_n\right)
\mu_0^{N+1}(d\mathbf z_{N+1}),
\end{multline}
where $\Sigma_i = \pm$ and
\begin{multline}\label{eq:psicompos}
\psi\left(t-\tau,\mathbf z_{n+1}^F(\tau,\mathbf z_{N+1},\mathbf K_N,\bm\Sigma_N),
\mathbf k_n,\bm\sigma_n\right)\\=
\varphi\left(z_1^F(t-\tau,\mathbf z_{n+1}^F(\tau,\mathbf z_{N+1},\mathbf K_N,\bm\Sigma_N),\mathbf k_n,\bm\sigma_n)\right)\\\times
 \chi_{A^\varepsilon_{\mathbf k_n\bm\sigma_n}(t-\tau)}
 \left(\mathbf z_{n+1}^F(\tau,\mathbf z_{N+1},\mathbf K_N,\bm\Sigma_N)\right)\;.
\end{multline}

Note that, in \eqref{eq:ttau3}, we have replaced the product of limits $\lim_{\varepsilon_q \to 0}$, $q=1,\cdots,n+1$ (running over the final points of the tree $\mathbf k_n$), with a single limit $\lim_{\varepsilon \to 0}$. The result is indeed independent on the precise way in which the regularization is removed in the different subtrees.

In the first factor in the right-hand side of (\ref{eq:psicompos}), we see a composition of two Boltzmann--Enskog forward flows: for the interval $[0,\tau]$ (determined by the ``bottom trees'' $\mathbf k_{m_1}^1,\ldots,\mathbf k_{m_{n+1}}^{n+1}$) and for the interval $[\tau,t]$ (determined by the ``top tree'' $\mathbf k_n$). We have 
\begin{equation*}
z_1^F(t-\tau,\mathbf z_{n+1}^F(\tau,\mathbf z_{N+1},\mathbf K_N,\bm\Sigma_N),\mathbf k_n,\bm\sigma_n)=z_1^F(t,\mathbf z_{N+1},\mathbf K_N,\bm\Sigma_N).
\end{equation*}

The characteristic function in (\ref{eq:psicompos})  means that, actually, the integration in (\ref{eq:ttau3}) takes place over initial configurations in $$A^\varepsilon_{\mathbf k^1_{m_1}\bm\sigma^1_{m_1}}(\tau)\times\ldots\times
A^\varepsilon_{\mathbf k^{n+1}_{m_{n+1}}\bm\sigma^{n+1}_{m_{n+1}}}(\tau)$$ leading to configurations in $A^\varepsilon_{\mathbf k_n\bm\sigma_n}(t-\tau)$ at the instant $\tau$, i.e., over the set
\begin{multline}\label{eq:setpreB}
\Big\{\mathbf z_{N+1}\in A^\varepsilon_{\mathbf k^1_{m_1}\bm\sigma^1_{m_1}}(\tau)\times\ldots\times
A^\varepsilon_{\mathbf k^{n+1}_{m_{n+1}}\bm\sigma^{n+1}_{m_{n+1}}}(\tau) \,\Big|\,\\
\mathbf z_{n+1}^F(\tau,\mathbf z_{N+1},\mathbf K_N,\bm\Sigma_N)\in 
A^\varepsilon_{\mathbf k_n\bm\sigma_n}(t-\tau)\Big\}.
\end{multline}
Consider another set:
\begin{equation}\label{eq:setB}
\{\mathbf z_{N+1}\in A^\varepsilon_{\mathbf K_N\bm\Sigma_N}(t) \,|\,
\mathbf z_{n+1}^F(\tau,\mathbf z_{N+1},\mathbf K_N,\bm\Sigma_N)\in 
A^\varepsilon_{\mathbf k_n\bm\sigma_n}(t-\tau)\}=
B^{\mathbf k_n\bm\sigma_nt-\tau;\varepsilon}_{\mathbf K_N\bm\Sigma_Nt}.
\end{equation}
Since, for a fixed $(\mathbf k_n,\bm\sigma_n)$, the ``bottom trees'' 
$$(\mathbf k_{m_1}^1\bm\sigma^1_{m_1},\ldots,\mathbf k^{n+1}_{m_{n+1}}\bm\sigma^{n+1}_{m_{n+1}})$$ and the joint tree $(\mathbf K_N,\bm\Sigma_N)$ are in a one-to-one correspondence, sets (\ref{eq:setpreB}) and (\ref{eq:setB}) coincide if $\varepsilon=0$. 

If $\varepsilon>0$, then these sets slightly differ. Namely, set (\ref{eq:setpreB}) does not demand the time separation between subsequent collisions of the same particle in a neighbourhood of the instant $\tau$. Consider, for example, the second composition of trees on Fig.~\ref{fig:sumtrees}, where $t_1\in[\tau,t]$ and $t_2\in[0,\tau]$. There is no demand of time separation between $t_1$ and $t_2$ in this composition: both $t_1$ and $t_2$ are allowed to be arbitrary close to $\tau$. However, by assumption, the $\mu_0^{N+1}$-measure of initial configurations leading to collisions in a $\varepsilon$-neighbourhood of $\tau$ tends to zero  for all triples $(N,\mathbf K_N,\bm\Sigma_N)$. Hence, in the limit $\varepsilon\to0$, the integrals over set (\ref{eq:setpreB}) and over set (\ref{eq:setB}) with respect to the measure $\mu_0^{N+1}$ coincide. 

Thus, the right-hand side of (\ref{eq:ttau}) can be written as
\begin{multline}\label{eq:ttau4}
\sum_{n=0}^\infty\sum_{m_1=0}^\infty\ldots\sum_{m_{n+1}=0}^\infty
\left(\frac\lambda{a^2}\right)^N
\sum_{\mathbf k_n\in\mathcal K_n}\sum_{\bm\sigma_n\in\{\pm \}^n}\\
\sum_{\mathbf k^1_{m_1}\in\mathcal K_{m_1}}
\sum_{\bm\sigma^1_{m_1}\in\{\pm \}^{m_1}}
\ldots
\sum_{\mathbf k^{n+1}_{m_{n+1}}\in\mathcal K_{m_{n+1}}}
\sum_{\bm\sigma^{n+1}_{m_{n+1}}\in\{\pm \}^{m_{n+1}}}
\prod_{i=1}^N\Sigma_i\lim_{\varepsilon\to0}\\
\int_{B^{\mathbf k_n\bm\sigma_nt-\tau;\varepsilon}_{\mathbf K_N\bm\Sigma_Nt}}
\varphi\left(z_1^F(\tau,\mathbf z_{N+1},\mathbf K_N,\bm\Sigma_N)\right)
\mu_0^{N+1}(d\mathbf z_{N+1}).
\end{multline}
The summation over 
$(m_1,\mathbf k^1_{m_1},\bm\sigma^1_{m_1},\ldots,
m_{n+1},\mathbf k^{n+1}_{m_{n+1}},\bm\sigma^{n+1}_{m_{n+1}})$ can be replaced by the summation over $(N,\mathbf K_N,\bm\Sigma_N)$, because, as we said before, they are in one-to-one correspondence for fixed $(n,\mathbf k_n,\bm\sigma_n)$. Hence, (\ref{eq:ttau4}) can be rewritten as
\begin{multline}\label{eq:ttau5}
\sum_{N=0}^\infty\left(\frac\lambda{a^2}\right)^N
\sum_{\mathbf K_N\in\mathcal K_N}\sum_{\bm\Sigma_N\in\{\pm \}^N}
\sum_{n=0}^\infty
\sum_{\mathbf k_n\in\mathcal K_n}\sum_{\bm\sigma_n\in\{\pm \}^n}
\prod_{i=1}^N\Sigma_i\\\times
\lim_{\varepsilon\to0}
\int_{B^{\mathbf k_n\bm\sigma_nt-\tau;\varepsilon}_{\mathbf K_N\bm\Sigma_Nt}}
\varphi\left(z_1^F(\tau,\mathbf z_{N+1},\mathbf K_N,\bm\Sigma_N)\right)
\mu_0^{N+1}(d\mathbf z_{N+1}).
\end{multline}
For every triple $(N,\mathbf K_N,\bm\Sigma_N)$,
\begin{equation*}
\bigcup_{n=0}^\infty
\bigcup_{\mathbf k_n\in\mathcal K_n}
\bigcup_{\bm\sigma_n\in\{\pm \}^n}
B^{\mathbf k_n\bm\sigma_nt-\tau;\varepsilon}_{\mathbf K_N\bm\Sigma_Nt}=
A^\varepsilon_{\mathbf K_N\bm\Sigma_N}(t),
\end{equation*}
hence the summation over $(n,\mathbf k_n,\bm\sigma_n)$ in (\ref{eq:ttau5}) gives the integral over $A^\varepsilon_{\mathbf K_N\bm\Sigma_N}(t)$:
\begin{multline*}
\int_{\mathbb R^6}\varphi(x_1,v_1)\mu_t(dx_1dv_1)=
\sum_{N=0}^\infty\left(\frac\lambda{a^2}\right)^N
\sum_{\mathbf K_N\in\mathcal K_N}
\sum_{\bm\Sigma_N\in\{\pm \}^N}\prod_{i=1}^N\Sigma_i\\ \times\lim_{\varepsilon\to0}
\int_{A^\varepsilon_{\mathbf K_N\bm\Sigma_N}(t)}
\varphi\left(z_1^F(t,\mathbf z_{N+1},\mathbf K_N,\bm\Sigma_N)\right)
\mu_0^{N+1}(d\mathbf z_{N+1}),
\end{multline*}
which coincides with (\ref{eq:weakr}).
\end{proof}

\section{Microscopic solutions of the Boltzmann--Enskog equation} \label{sec:MsBEE}

\subsection {Main result}

We are going to prove that a measure of form (\ref{eq:microsol}) is a weak series solution of the Boltzmann--Enskog equation, i.e., satisfies (\ref{eq:weakr}). We assume that the initial configuration $\mathbf z^0=(z^0_1,\ldots,z^0_N)$ is such that  configurations $\left(z^0_i\right)_{i\in\mathcal I_n}$, for all non-empty
$\mathcal I_n\subset\{1,\ldots,N\}$, lead to a well-defined hard-sphere motion and do not lead to grazing collisions. Also, for simplicity, let us assume that the initial configuration $\mathbf z^0$ does not lead to simultaneous collisions of different pairs of hard spheres. As previously recalled, the set of initial configurations satisfying these conditions has full measure.

\begin{theorem}\label{th}
Let $\lambda=Na^2$. Under the above assumptions on $\mathbf z^0$, the empirical measure (\ref{eq:microsol}) is a weak series solution of the Boltzmann--Enskog equation, i.e., it satisfies (\ref{eq:weakr}).
\end{theorem}

\begin{proof}
Denote by $\mathrm T_n(s)$, $s\in [0,t]$ the $n$-particle hard sphere flow. As before, denote $\mathbf z(s)=\mathrm T_N(s)\mathbf z^0$. Let $t$ be such that the configurations $\left(z^0_i\right)_{i\in\mathcal I_n}$ do not lead to collisions at the instant $t$. 

Let us partition the interval $[0,t)$ in a sequence of intervals $[\theta_j,\theta_{j+1})$, $j=0,\ldots,S-1$, $0=\theta_0<\theta_1<\ldots<\theta_S=t$ with the following properties:

\begin{enumerate}[(i)]
\item In each interval $(\theta_j,\theta_{j+1})$, at most one collision in the hard sphere dynamics starting from $\mathbf z^0$ occurs. There are no collision at the instants $\theta_j$, $j=1,\ldots,S$.

\item Let the particles~$l$ and~$m$ collide in the interval $[\theta_j,\theta_{j+1})$. Then both flows 
$$\mathrm T_{N-1}(s)\left(z_i(\theta_j)\right)_{i\in\{1,\ldots,l-1,l+1,\ldots,N\}},
\quad s\in[0,\theta_{j+1}-\theta_j],$$
and 
$$\mathrm T_{N-1}(s)\left(z_i(\theta_j)\right)_{i\in\{1,\ldots,m-1,m+1,\ldots,N\}},
\quad s\in[0,\theta_{j+1}-\theta_j],$$ 
are free.
\end{enumerate}

Due to the semigroup property (Proposition~\ref{prop:semigroup}), it is sufficient to prove (\ref{eq:weakr}) for a single interval, say, $[0,\theta_1)$. 

The left-hand side of (\ref{eq:weakr}) is, for all $s \in [0,\theta_1]$
\begin{equation}\label{eq:lhs}
\int_{\mathbb R^6}\varphi(x_1,v_1)\mu_s(dx_1dv_1)=
\frac1N\sum_{i=1}^N\varphi(x_i(s),v_i(s)).
\end{equation}
Denote the terms on the right-hand side of (\ref{eq:weakr}) as
$\sum_{n=0}^\infty \Phi_n.$ The $n=0$ term is 
\begin{equation*}
\Phi_0=\int_{\mathbb R^6}\varphi(x_1+v_1s,v_1)\mu_0(dx_1dv_1)=
\frac1N\sum_{i=1}^N\varphi(x_i^0+v_i^0s,v_i^0).
\end{equation*}

If there are no collisions in the interval $[0,s)$, then $x_i(s)=x_i^0+v_i^0s$, $v_i(s)=v_i^0$, and $\Phi_0$ coincides with (\ref{eq:lhs}). Moreover, the higher-order terms vanish since $A^\varepsilon_{\mathbf k_n\bm\sigma_n}(s)$ does not intersect with $\supp \mu_0=\{z_i^0\}_{i=1}^N$ in this case. So, (\ref{eq:microsol}) obviously satisfies (\ref{eq:weakr}) in this interval.

Let now particles~$l$ and~$m$ collide at an instant $\tau\in(0,\theta_1)$. By our hypotheses, all the other particles move freely.
Then
\begin{eqnarray*}
x_l(s)&=&x_l^0+v_l^0\tau+v'_l(s-\tau),\\
x_m(s)&=&x_m^0+v_m^0\tau+v'_m(s-\tau),\\
x_i(s)&=&x_i^0+v_i^0s,\quad i\neq l,m,
\end{eqnarray*}
where $v'_l$ and $v'_m$ are related to $v_l^0$ and $v_m^0$ by (\ref{eq:coll}).
Considering the first-order terms, we have $\mathbf k_1=1$ (the only possibility) and
$$
A^\varepsilon_{\mathbf k_1\bm\sigma_1}\cap\supp\mu_0^2=
\left\lbrace
(z_l^0,z_m^0),(z_m^0,z_l^0)
\right\rbrace
$$
for every $\sigma_1=\pm $. The first-order term in the right-hand side of (\ref{eq:weakr}) is then
\begin{equation*}
\begin{split}
\Phi_1=\:&N\sum_{\sigma_1\in\{\pm \}}
\sigma_1\lim_{\varepsilon\to0}
\int_{A^\varepsilon_{\mathbf k_1\bm\sigma_1}(s)}
\varphi(z_1^F(s,z_1,z_2,k_1,\sigma_1))
\mu^2_0(dz_1dz_2)\\
=\:&\frac1N\left[
\varphi(z_1^F(s,(z_l^0,z_m^0),\mathbf k_1,+))
+\varphi(z_1^F(s,(z_m^0,z_l^0),\mathbf k_1,+))\right.\\
&\:\left.-\varphi(z_1^F(s,(z_l^0,z_m^0),\mathbf k_1,-))
-\varphi(z_1^F(s,(z_m^0,z_l^0),\mathbf k_1,-))
\right]\;.
\end{split}
\end{equation*}
In the case $\bm\sigma_1=+$, the Boltzmann--Enskog forward flow coincides with the hard sphere dynamics, i.e., 
\begin{equation*}
\begin{split}
z_1^F(s,(z_l^0,z_m^0),\mathbf k_1,+)&=z_l(s)\:\:=(x_l^0+v_l^0\tau+v'_l(s-\tau),v'_l),\\
z_1^F(s,(z_m^0,z_l^0),\mathbf k_1,+)&=z_m(s)=(x_m^0+v_m^0\tau+v'_m(s-\tau),v'_m).
\end{split}
\end{equation*}
In the case $\bm\sigma_1=-$, the Boltzmann--Enskog forward flow for the particle~1 coincides with the free dynamics of the particle~1, i.e.,
\begin{equation*}
\begin{split}
z_1^F(s,(z_l^0,z_m^0),\mathbf k_1,-)&=(x_l^0+v_l^0s,v_l),\\
z_1^F(s,(z_m^0,z_l^0),\mathbf k_1,-)&=(x_m^0+v_m^0s,v_m).
\end{split}
\end{equation*}
Hence, 
\begin{equation*}
\Phi_1=
\frac1N\big[\varphi(x_l(s),v_l(s))
+\varphi(x_m(s),v_m(s))
-\varphi(x_l^0+v_l^0s,v^0_l)
-\varphi(x_m^0+v_m^0s,v^0_m)
\big],
\end{equation*}
so that $\Phi_0+\Phi_1$ is equal to (\ref{eq:lhs}).  

The higher-order terms vanish since $A^\varepsilon_{\mathbf k_n\bm\sigma_n}(s)$ for $n\geq2$ does not intersect with $\supp \mu_0$. Consider, for example, $n=2$. Configurations like $(z_i^0,z_j^0,z_r^0)$ with $i\neq j\neq r$ for $\bm\sigma_2=(+,+)$ do not lead to two subsequent collisions, since the Boltzmann--Enskog forward flow coincides with the hard sphere dynamics for this choice of $\bm\sigma_2$. But, by property~(i) of the intervals, hard spheres suffer at most one collision on the interval $[0,\theta_1)$. The same configurations for $\bm\sigma_2=(+,-), (-,+), (-,-)$ do not lead to two collisions by property~(ii): neither particle~$l$ ignoring particle~$m$ collides with another particle nor particle~$m$ ignoring particle~$l$ collides with  another particle; other particles move freely. Configurations like $(z_l^0,z_m^0,z_m^0)\in\supp\mu_0^3$ could be interpreted as configurations leading to two collisions: the first particle collides with the second and the third particles at the same instant $\tau$. However, by definition of 
$A^\varepsilon_{\mathbf k_n\bm\sigma_n}(s)$, the subsequent collisions of the same particle should be $\varepsilon$-separated in time.

Thus, (\ref{eq:microsol}) satisfies (\ref{eq:weakr}) in the interval $[0,\theta_1)$ in the case of a single collision in this interval. This concludes the proof of the theorem. 

\end{proof}

\begin{remark}
Note that Proposition~\ref{prop:converg}, adapted to the present context, guarantees the convergence of series (\ref{eq:weakr}) only for small $\lambda$. However, in case of microscopic solutions, the series has only a finite number of non-vanishing terms, by construction, so that it converges for arbitrary large $\lambda=Na^2$.
\end{remark}

\subsection{Recollisions vs.\,contractions} \label{subsec:RvC}

There is another and more natural way to describe the hard-sphere dynamics for a microscopic state in terms of a series expansion.

The $j$-particle marginals associated to a time evolving particle configuration
$\bar {\mathbf z}_N(t) =\mathrm T_N(t)\bar {\mathbf z}_N(0)=\{\bar z_1(t), \cdots, \bar z_N (t)\}$
are defined by
\begin{equation}
\label{empmar}
\Delta_j  (\mathbf z_j,t )= \frac 1{N(N-1) \cdots (N-j+1)}\sum_{\substack{i_1,\cdots,i_j \\ i_a \neq i_b}} \prod_{s=1}^j   \delta (z_s- \bar z_{i_s} (t) )
\end{equation}
where $ \mathbf z_j = (z_1 \cdots z_j )$. 

For $j=1$ we recover the time-evolved empirical distribution \eqref {eq:microsol}. 

In  \cite{PulvirSimon}   it has been proved that \eqref{empmar} satisfies the following series expansion (BBGKY hierarchy)

\begin{multline}\label{eq:hierint}
\Delta_j(t)=\mathcal S^{int}_{1\ldots j}f_{0,j}+\sum_{n=1}^{N-j}
a^{2n} (N-j)(N-j-1) \dots (N-j-n+1)\\
\int_0^tdt_1\int_0^{t_1}dt_2\ldots\int_0^{t_{n-1}}dt_n\\
\mathcal S^{int}_{1\ldots j}(t-t_1)\mathcal C_{j+1}
\mathcal S^{int}_{1\ldots j+1}(t_1-t_2)\ldots
\mathcal C_{j+n}\mathcal S^{int}_{1\ldots j+n}(t_n)\Delta_{0,j+n},
\end{multline}
where 
$$
\mathcal S^{int}_{1\ldots j}(t) F(\mathbf z_j)=F( \mathrm T_j(-t) \mathbf z_j)=F(  \mathbf z_j (-t))
$$
and $\mathrm T_j$ is the $j$-particle interacting flow.

The validity of the above expansion, for $j=1$ can be compared with the one of the Boltzmann--Enskog equation,
which we rewrite here as
\begin{multline}\label{eq:hiermeas}
\mu_t=\mathcal S_{1} (t) \mu_0+\sum_{n=1}^\infty\lambda^n
\int_0^tdt_1\int_0^{t_1}dt_2\ldots\int_0^{t_{n-1}}dt_n\\
\mathcal S_{1}(t-t_1)\mathcal C_{2}
\ldots
\mathcal C_{1+n}\mathcal S_{1\ldots 1+n}(t_n)\mu_{0,1+n},
\end{multline}
where $\mu_{0,j}=\mu_0^{\otimes j}$. 
Eq.\,\eqref {eq:hiermeas} has to be interpreted according to the notion of weak series solution discussed in this paper.

The right-hand sides of  \eqref {eq:hierint} for $j=1$ and \eqref{eq:hiermeas} must
coincide, being identical the left hand sides. How is this possible? 

First, observe that
$$
\mu_{0,j}=\frac {N(N-1)(N-2) \dots (N-j+1)}{N^j}  \Delta_{0,j}.
$$
Note that the above identity holds on the physical phase space only, namely on the set of configurations
for which $|x_k-x_\ell |>a $ for all $k \neq \ell $. Furthermore we take into account that 
$\lambda=a^2N$ and that the the sum $\sum^{N-1}$ can be replaced by $\sum^{\infty}$ by adding vanishing terms.
Then we conclude that the two series are identical but for the fact that 
$\mathcal S^{int} (t) \neq \mathcal S(t)$. More precisely, the backward flow  $\mathcal S^{int} (t)$ describes also the collisions among the particles already created. Such interactions are usually called `recollisions'. In contrast, $\mathcal S(t)$ is just the free flow which implies that, once two particles are created, if  they arrive at distance $a$ (necessarily with outgoing velocities), they go ahead freely by backward overlapping. 

More precisely, \cite{PulvirSimon} proves the validity of a representation of the form:

\begin{multline}\label{eq:weaktrue}
\int_{\mathbb R^6}\varphi(x_1,v_1)\mu_t(dx_1dv_1)=
\sum_{n=0}^\infty\left(\frac\lambda{a^2}\right)^n
\sum_{\mathbf r_n }
\sum_{\bm\sigma_n\in\{\pm \}^n}\prod_{i=1}^n\sigma_i
\\\times\int_{\bar A_{\mathbf r_n\bm\sigma_n}(t)}\varphi\left(z_1^{FI}(t,\mathbf z_{n+1},\mathbf r_n,\bm\sigma_n)\right)
\mu_0^{\otimes (n+1)}(d\mathbf z_{n+1}),
\end{multline}
where  $z_1^{FI}(t,\mathbf z_{n+1},\mathbf r_n,\bm\sigma_n))$ is the interacting forward flow, taking into account the recollisions.
Here there is no necessity of time separation and, as we have already illustrated, the sum over the fully ordered trees can be replaced by the sum over the partially ordered trees.

\begin{remark}
The above equation must be further clarified (see \cite{PulvirSimon}).
In fact, in contrast with the Boltzmann--Enskog flow, the corresponding $(\mathbf r_n,\bm\sigma_n)$-dependent map
(`interacting backward flow')
\begin{equation}\label{eq:bflowI}
(x_1,v_1,\mathbf t_n,\bm\omega_n,\mathbf v_n)\longmapsto \bm\zeta^{BI} (0),
\end{equation}
of which $\bar A_{\mathbf r_n\bm\sigma_n}(t)$ is the image,
is not globally invertible, but only locally. However \eqref{eq:weaktrue} makes sense if we assume that $\varphi$ has small enough support. We shall make this assumption in the discussion that follows.
\end{remark}

The two representations   \eqref {eq:weaktrue} and \eqref {eq:weakr}
provide indeed the same result.  

Consider the following simple example with $N=4$. The tree associated to \eqref {eq:weaktrue} is given by $ {\bf r}_3= \{1,1,2 \}$ while  the one associated to the Boltzmann--Enskog expansion is ${\bf k}_6=\{2,1,1,1,0,0 \}$, see Fig.\,\ref{fig:esComp}.

\begin{figure}[h]
\centering
\includegraphics[width=10cm]{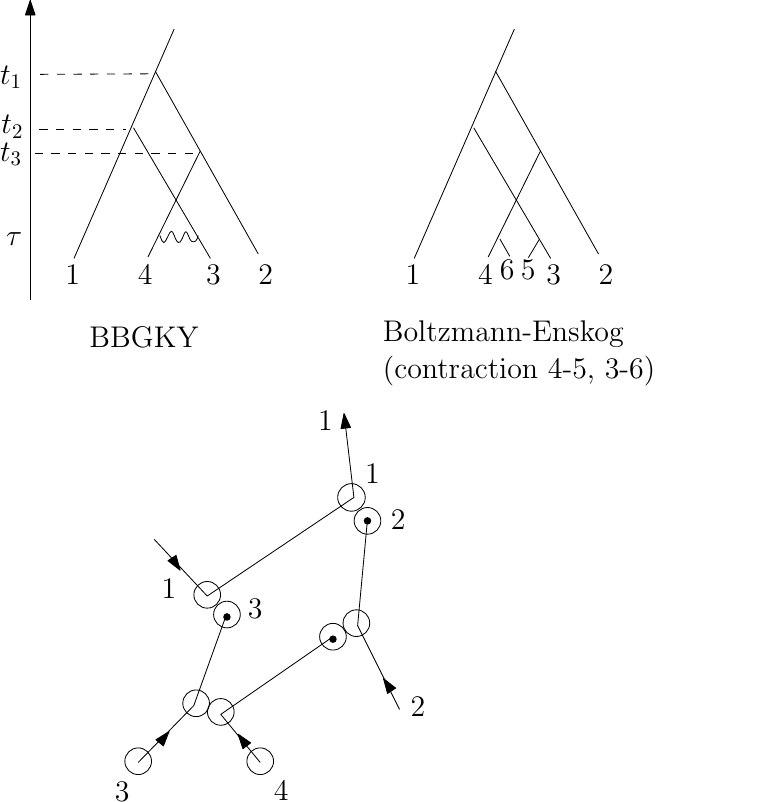}
\caption{} \label{fig:esComp}
\end{figure}

The wavy line of the first tree denotes that particles $3$ and $4$ recollide at some time $\tau \in (0,t_3)$ according, say, to the hard-sphere dynamics in the figure.

Such a recollision can be described as well in terms of the Boltzmann--Enskog flow, by a creation of two (fictitious) particles $5$ and $6$. Recalling that $ \frac {\lambda}{a^2}=N=4$, the $5$-th order term of \eqref{eq:weakr} is
then equal to the $3$-th order term of  \eqref {eq:weaktrue}. Indeed the former is
\begin{multline*}
\frac 14 \varphi (\bar z_1(t)) \int \prod_{i=1}^6 dz_i  
\delta (z_1-\bar z_1(0)) \delta (z_2-\bar z_2(0)) \delta (z_3-\bar z_3(0)) \delta (z_4-\bar z_4(0)) \\ 
\times\delta (z_5-\bar z_4(0)) \delta (z_6-\bar z_3(0))=
\frac 14 \varphi (\bar z_1(t))
\end{multline*}
while the second is
\begin{equation*}
\frac 14 \varphi (\bar z_1(t)) \int \prod_{i=1}^4 dz_i 
\delta (z_1-\bar z_1) \delta (z_2-\bar z_2)  \delta (z_3-\bar z_3) \delta (z_4-\bar z_4)=
\frac 14 \varphi (\bar z_1(t))\;.
\end{equation*}

Summarizing we see that the integration in $\prod_{i=1}^6 dz_i $ follows by a reduction of the computation at time zero by means of the map \eqref{map}  associated to the Boltzmann--Enskog backward flow, while the integration  $\prod_{i=1}^4 dz_i $  follows by the corresponding map associated to the interacting backward flow.

This example shows also that a full time separation cannot work as regularization, because it cannot describe recollisions which, in the Boltzmann--Enskog expansion, are given by simultaneous creations of two particles, contracted with previously existing particles of the backward flow.

\begin{remark}
The previous argument shows that Eq.\,\eqref {eq:weaktrue} could be derived by  assuming \eqref {eq:weakr} and vice versa.
\end{remark}

\subsection{Conclusions}

The presented way of giving a rigorous sense to the microscopic solutions of the Boltzmann--Enskog equation appears to be the most natural one, in comparison to the previously proposed variants \cite{TrushPad, TrushQP, TrushMIAN, TrushKRM}, since it does not involve regularizations of delta functions, but gives direct sense to weak solutions of the Boltzmann--Enskog equation by means of a notion of series solution. 
Here we have introduced a series expansion that involves only partial chronological ordering of collisions. The result is recovered by the time separation of collisions corresponding to this partial order.
Thus, the Boltzmann--Enskog equation, which is known to describe irreversible dynamics and entropy production, contains also solutions corresponding to the reversible microscopic dynamics of hard spheres. 

Formula (\ref{eq:weakr})  for weak series solutions is not very handable for practical purposes, but reveals a relation of solutions of the Boltzmann--Enskog equation with the hard sphere dynamics.

\section*{Acknowledgments} The work of A.T. was supported by the grant of the President of the Russian Federation (project MK-2815.2017.1). S.S.\,acknowledges support from DFG grant 269134396.

\medskip
Received xxxx 20xx; revised xxxx 20xx.
\medskip


\begin{thebibliography}{99}

\bibitem{Alexander} (MR2625918)
 \newblock R.K. Alexander, 
  \newblock The infinite hard sphere system, 
  \newblock  Ph.D thesis, Dep. of Mathematics, University of California at Berkeley, 1975.

\bibitem{ArkCerc89} (MR1017064)
\newblock L. Arkeryd and C. Cercignani, 
\newblock On the convergence of solutions of the Enskog equation to solutions of the Boltzmann equation, 
\newblock \emph{Comm. PDE}, \textbf{14} (1989), 1071--1090.

\bibitem{ArkCerc} (MR1063185)
\newblock L. Arkeryd and C. Cercignani, 
\newblock Global existence in $L_1$ for the Enskog equation and convergence of the solutions to solutions of the Boltzmann equation. 
\newblock \emph{J. Stat. Phys.}, \textbf{59} (1990), 845--867. 

\bibitem{BelLach88} (MR0952754)
\newblock N. Bellomo and M. Lachowicz, 
\newblock On the asymptotic equivalence between the Enskog and the Boltzmann equations. 
\newblock \emph{J. Stat. Phys.}, \textbf{51} (1988), 233--247. 

\bibitem{BGSR13} (MR3455156)
\newblock T. Bodineau, I. Gallagher and L. Saint--Raymond,
\newblock The Brownian motion as the limit of a deterministic system of hard--spheres. 
\newblock \emph{Inventiones}, \textbf{203} (2016), 493--553. 

\bibitem{BGSR15} (MR3625187)
\newblock T. Bodineau, I. Gallagher and L. Saint--Raymond,
\newblock From hard sphere dynamics to the Stokes--Fourier equations: an $L^2$ analysis of the
Boltzmann--Grad limit. 
\newblock \emph{Annals PDE}, \textbf{3} (2017). 

\bibitem{Bodi} 
\newblock T. Bodineau, I. Gallagher, L. Saint-Raymond and S. Simonella,
\newblock One-sided convergence in the Boltzmann-Grad limit,
\newblock \emph{Ann.\;Fac.\;Sci.\;Toulouse Math.} (to appear).

\bibitem{Bogol46} 
\newblock N. N. Bogoliubov, 
\newblock \emph{Problems of Dynamic Theory in Statistical Physics}, 
\newblock Gostekhizdat, Moscow--Leningrad, 1946; North-Holland, Amsterdam, 1962; Interscience, New York, 1962.

\bibitem{Bogol75} (MR0468958)
\newblock N. N. Bogolyubov, 
\newblock Microscopic solutions of the Boltzmann--Enskog equation in kinetic theory for elastic balls, 
\newblock  \emph{Theor. Math. Phys.}, \textbf{24} (1975), 804--807. 

\bibitem{BogBog} (MR0785359)
\newblock N. N. Bogolubov and N. N. (Jr.) Bogolubov, 
\newblock \emph{Introduction to Quantum Statistical Mechanics}, 
\newblock Nauka, Moscow, 1984; World Scientific, Singapore, 2010.

\bibitem{CGP} (MR1472233)
\newblock C. Cercignani, V. I. Gerasimenko and D. Y. Petrina, 
\newblock \emph{Many-Particle Dynamics and Kinetic Equations},
\newblock Kluwer Academic Publishing, Dordrecht, 1997.

\bibitem{CIP} (MR1307620)
\newblock C. Cercignani, R. Illner and M. Pulvirenti, 
\newblock \emph{The Mathematical Theory of Dilute Gases},
\newblock Springer--Verlag, New York, 1994.

\bibitem{De17}
\newblock R. Denlinger, 
\newblock The propagation of chaos for a rarefied gas of hard spheres in the whole space,
\newblock preprint, \arXiv{1605.00589}.

\bibitem{GST} (MR3157048)
\newblock I. Gallagher, L. Saint Raymond and B. Texier, 
\newblock \emph{From Newton to Boltzmann: Hard Spheres and Short-Range Potentials}, 
\newblock Z\"urich Adv. Lect. in Math. Ser. \textbf{18}, EMS, 2014, and erratum to Chapter 5.

\bibitem{GerasGap} (MR2972447)
\newblock V. I. Gerasimenko and I. V.  Gapyak, 
\newblock Hard sphere dynamics and the Enskog equation,
\newblock  \emph{Kinet. Relat. Models}, \textbf{5} (2012), 459--484. 

\bibitem{Lanford} (MR0479206)
\newblock O. E. Lanford, 
\newblock Time evolution of large classical systems, 
\newblock \emph{Lect. Notes Phys.}, {\bf 38} (1975), 1--111.

\bibitem{Pulvir} (MR1461101)
\newblock M. Pulvirenti, 
\newblock On the Enskog hierarchy: analiticity, uniqueness and derivability by particle systems, 
\newblock \emph{Rend. Circ. Mat. Palermo} 2 (1996), 529--542.

\bibitem{PSS} (MR3190204)
\newblock M. Pulvirenti, C. Saffirio and S. Simonella, 
\newblock On the validity of the Boltzmann equation for short-range potentials,
\newblock \emph{Rev. Math. Phys.}, \textbf{26} (2014), 1--64.

\bibitem{PulvirSimon} (MR3409818)
\newblock M. Pulvirenti and S. Simonella, 
\newblock On the evolution of the empirical measure for the hard-sphere dynamics,
\newblock \emph{Bull. Inst. Math. Academia Sinica}, \textbf{10} (2015), 171--204.

\bibitem{PulvirSimonBG} (MR3608289)
\newblock M. Pulvirenti and S. Simonella, 
\newblock The Boltzmannn-Grad limit of a hard sphere system: analysis of the correlation error,
\newblock \emph{Inventiones}, {\bf 207} (2017), 1135-1237.

\bibitem{Simonella} (MR3207735)
\newblock S. Simonella, 
\newblock Evolution of correlation functions in the hard sphere dynamics,
\newblock \emph{J. Stat. Phys.}, \textbf{155} (2014) 1191--1221.

\bibitem{Spohn}
\newblock H. Spohn, 
\newblock \emph{Large-Scale Dynamics of Interacting Particles}, 
\newblock Springer, Berlin, 1991.

\bibitem{TrushPad} (MR2915625)
\newblock A. S. Trushechkin, 
\newblock Derivation of the particle dynamics from kinetic equations, 
\newblock \emph{\textit{p}-Adic, Ultrametric Analysis and Applications}, \textbf{4} (2012), 130--142.

\bibitem{TrushQP} (MR2840732)
\newblock A. S. Trushechkin, 
\newblock Functional mechanics and kinetic equations, 
\newblock \emph{QP--PQ: Quantum probability and White Noise Analysis}, \textbf{30} (2013), 339--350.

\bibitem{TrushMIAN} (MR3479999)
\newblock A. S. Trushechkin, 
\newblock Microscopic solutions of the Boltzmann--Enskog equation and the irreversibility problem, 
\newblock \emph{Proc. Steklov Inst. Math.}, \textbf{285} (2014), 251--274.

\bibitem{TrushKRM} (MR3317580)
\newblock A. S. Trushechkin, 
\newblock Microscopic and soliton-like solutions of the Boltzmann--Enskog and generalized Enskog equations for elastic and inelastic hard spheres, 
\newblock \emph{Kinetic and Relat. Models}, \textbf{7} (2014), 755--778.

\end{thebibliography}
\end{document}